\documentclass[letterpaper, 10 pt, conference]{ieeeconf}  %

\IEEEoverridecommandlockouts                              %

\usepackage{cite}
\usepackage{amsmath,amssymb,amsfonts}
\usepackage{commath}
\usepackage{bm}
\usepackage{algpseudocode}
\usepackage{algorithm}
\usepackage{graphicx}
\graphicspath{{figs_VSS2026/}}
\usepackage{subcaption}
\usepackage{textcomp}
\usepackage{xcolor}
\usepackage{multirow}
\usepackage{float}
\usepackage{booktabs}
\usepackage{stackengine}
\usepackage{etoolbox}
\usepackage{mathtools}
\usepackage{optidef}

\newtheorem{assumption}{Assumption}

\newtheorem{theorem}{Theorem}
\newtheorem{prop}{Proposition}

\newtheorem{remark}{Remark}

\usepackage[utf8]{inputenc}
\DeclareUnicodeCharacter{2212}{-}
\newcommand{\RNum}[1]{\uppercase\expandafter{\romannumeral#1\relax}}

\title{\LARGE \bf
	Robust Hybrid-Gain Control with Prescribed Transient Performance
}

\author{Amit Shivam$^{1}$, Kiran Kumari$^{2}$, and Fernando A.C.C. Fontes$^{1}$ 
	\thanks{$^{1}$Amit Shivam and Fernando A.C.C. Fontes are with SYSTEC, ARISE, Faculty of Engineering, Universidade do Porto, Porto, Portugal, Email:
		{\tt\small amit@upwind.pt}, and {\tt\small faf@fe.up.pt}}
	\thanks{$^{2}$Kiran Kumari is with Department of Electrical Engineering, Indian Institute of Science, Bengaluru, India, Email: {\tt\small kirank@iisc.ac.in}}
    \thanks{This research has been supported by FCT/MCTES (PIDDAC), through project 2022.02801.PTDC-UPWIND-ATOL (https://doi.org/10.54499/2022.02801.PTDC) and grant 2021.07346.BD.}
}

\begin{document}
	
	\maketitle
	\thispagestyle{empty}
	\pagestyle{empty}

\begin{abstract}
   This paper develops a hybrid-gain robust control framework
for first- and second-order nonlinear systems subject to
bounded matched disturbances and prescribed transient
performance constraints. A common error-function
transformation incorporates time-varying tracking-error
bounds, while a hybrid-gain-based feedback control law regulates convergence
and disturbance rejection. Finite-time convergence and
forward invariance of the prescribed-performance funnel
are established for first-order systems. For second-order
systems, an auxiliary velocity error enables exponential
convergence and funnel invariance without assuming
a priori boundedness of the transformed state.
An initial-condition-dependent bound further provides
a sufficient condition for realising the ideal control
under actuator limits. Numerical comparisons under common actuator limits and peak-control input demonstrate the trade-off between transient constraint satisfaction and control effort. At equal peak control, tuned
first-order baselines track better at the cost of large transient energy. In contrast, the proposed
design enforces the funnel guarantee without retuning. In the second-order study, the proposed controller achieves 
lower integral error and control energy at equal peak control, while under
actuator saturation it can leave the funnel slightly. Initial-condition tests further demonstrate the effect of actuator saturation, showcasing ideal funnel-invariance guarantees from physically constrained performance.
\end{abstract}
\section{Introduction}
Finite-time sliding-mode control (FTSMC) provides
robust convergence under matched disturbances, but
finite-time stability alone does not constrain the
tracking error throughout its transient evolution.
A trajectory may reach the equilibrium in finite
time while exceeding an admissible error bound
before convergence. This distinction becomes
important when transient excursions must be
restricted independently of the settling-time
requirement.
Terminal and nonsingular terminal sliding-mode
controllers have been developed to obtain finite-time
convergence while addressing singularity and
robustness issues~\cite{galicki2015finite,du2011finite,li2017finite,venkataraman1993control,yu2002fast, YU2005,feng2002non,feng2013nonsingular}. Their stability properties are
supported by finite-time analysis
\cite{BhatandBernstein2000,basin2019finite,Yu2021}, while prescribed-time
control provides related convergence guarantees
\cite{song2023prescribed}. These results, however, do not
generally impose a specified time-varying error
envelope. Most finite-time SMC designs rely on their implicit and largely uncontrolled transient behavior. While finite-time convergence guarantees that the state reaches the equilibrium within a bounded time, it does not specify how the state evolves during the transient phase. In practice, this often leads to undesirable behaviors such as large overshoots, aggressive control actions, or violation of constraints before convergence is achieved. Such transient excursions are especially problematic in actuator-limited systems and safety-critical applications, where exceeding predefined bounds—even temporarily—may not be feasible in real-world applications. 

Prescribed-performance control (PPC) addresses
transient-error constraints through a time-varying
performance function and an associated error
transformation. The resulting formulations have
been developed for feedback-linearizable and
uncertain nonlinear systems~\cite{Bechlioulis2008,BECHLIOULIS2014},
including adaptive finite-time designs~\cite{TRAN2020}.
Consequently, this work explores
 how a prescribed transformation can be
incorporated into a hybrid-gain robust controller
while establishing convergence directly in the
transformed coordinates and quantifying the
associated control requirements.
Recent work on hybrid-gain based finite-time sliding mode control (HG-FTSMC)\cite{shivamArxiv2024}(non-PPF) combines
constant and nonlinear gain components to regulate
convergence and control effort, and can significantly improve convergence behaviour with a reduced control effort. However, the design has so far been developed without explicit consideration of prescribed transient envelopes, leaving a gap between robust finite-time stabilization and guaranteed transient performance.
Motivated by these limitations, this paper proposes a PPF-aware HG-FTSMC
for scalar systems with matched disturbances. The contributions are summarised as follows:
\begin{enumerate}
\item A prescribed-performance-aware hybrid-gain
controller for first-order systems, with a
transformed-domain proof of funnel invariance,
finite-time convergence, and an explicit
reaching-time bound.

\item A virtual-velocity-based extension to
second-order systems, establishing exponential
convergence and funnel invariance without
assuming a priori boundedness of the transformed
error.

\item Initial-condition-dependent control bounds
that provide sufficient conditions for realizing
the ideal feedback under actuator constraints.

\item Numerical comparisons under common actuator
limits and peak-control input, including
several initial conditions for robustness tests to analyse
key controller performance metrics.
\end{enumerate}

The remainder of this paper is organized as follows. The problem formulation is discussed
in Section~\ref{sec:problem}. In Section~\ref{sec:ppf aware controller design}, the prescribed-performance-aware hybrid-gain controller design scheme
will be derived, followed
by numerical simulation
results in Section~\ref{sec:simulations}. Finally, Section~\ref{sec:conclusion} concludes the paper.

\section{Problem Formulation}\label{sec:problem}
Consider the following two classes of scalar control-affine systems:
\begin{align}
 \dot x &= f_1(x,t)+g_1(x,t)u+d_1(t), \label{eq:first_plant}\\
 \ddot x &= f_2(x,\dot x,t)+g_2(x,\dot x,t)u+d_2(t).\label{eq:second_plant}
\end{align}
Here $x\in\mathbb R$ is the tracking error relative to a constant zero reference, $u\in\mathbb R$ is the control input, and $d_i$ denotes matched uncertainty. The functions $f_i$ and $g_i$ are known nominal dynamics; the same notation is used for each system only where no ambiguity arises.

\begin{assumption}
  For $i\in\{1,2\}$, $f_i$ and $g_i$ are continuous in time and locally Lipschitz in their state arguments on the domain of interest. The input gain has known sign and satisfies $|g_i|\ge g_{i,\min}>0$. Each disturbance $d_i$ is measurable and essentially bounded, with $|d_i(t)|\le D_i$ almost everywhere, for a known $D_i\ge0$. State measurements and the nominal dynamics needed for feedback cancellation are available.  
\end{assumption}

The desired performance envelope is
\begin{equation}\label{eq:rho}
 \rho(t)=(\rho_0-\rho_\infty)e^{-\lambda t}+\rho_\infty,
 \quad \rho_0>\rho_\infty>0,\quad\lambda>0,
\end{equation}
with prescribed admissible region
\begin{equation}\label{eq:funnel}
 \mathcal F(t)=\{x\in\mathbb R:|x|<\rho(t)\}.
\end{equation}
\begin{assumption}
    The initial position error satisfies $|x(0)|<\rho_0$. For the second-order system, $\dot x(0)$ is finite. No a priori bound on the transformed state is assumed.
\end{assumption}

\emph{Control problem:} Design control laws for \eqref{eq:first_plant} and \eqref{eq:second_plant} that (i) render \eqref{eq:funnel} forward invariant for every admissible disturbance, (ii) guarantee convergence of the tracking error to zero, and (iii) characterize the resulting convergence rates and control requirements. 

\section{Prescribed-Performance-Aware Hybrid-Gain Controller Design}\label{sec:ppf aware controller design}
This section discusses the main results with error transformation to enforce prescribed performance, followed by controller design for first-order\eqref{eq:first_plant} and second-order\eqref{eq:second_plant} system dynamics.
\subsection{Error transformation and hybrid gain}\label{subsec:common}
For either system, define the normalized error $y=x/\rho(t)$ and the diffeomorphic transformation
\begin{equation}\label{eq:transform}
 \xi=\operatorname{erf}^{-1}(y),\qquad x=\rho(t)\operatorname{erf}(\xi),
\end{equation}
valid on $|x|<\rho(t)$. Since $d\operatorname{erf}(\xi)/d\xi=2e^{-\xi^2}/\sqrt\pi$, differentiation yields
\begin{align}
 \dot x&=\dot\rho\operatorname{erf}(\xi)+\frac{2\rho}{\sqrt\pi}e^{-\xi^2}\dot\xi,\label{eq:xdot_transform}\\
 \dot\xi&=\chi(t,\xi)\left(\dot x-\frac{\dot\rho}{\rho}x\right),\qquad
 \chi(t,\xi)=\frac{\sqrt\pi e^{\xi^2}}{2\rho(t)}>0.\label{eq:xi_dot_common}
\end{align}
In particular, a bounded $\xi$ implies strict funnel satisfaction, but this boundedness will be established from the closed-loop dynamics.
For $r\ge0$, introduce the hybrid gain
\begin{equation}\label{eq:gain}
 G(r)=\begin{cases}
 k_d+a r^\gamma+b r^\alpha,&0\le r\le\varepsilon,\\
 k_0+k_1\dfrac{r^\gamma}{\varepsilon_0^\gamma+r^\gamma},&r>\varepsilon,
 \end{cases}
\end{equation}
where $a,b,k_0,k_1,k_d,\varepsilon,\varepsilon_0>0$, $0<\gamma<1<\alpha$, and $\varepsilon\le\varepsilon_0$. Parameters are chosen to satisfy continuity at $r=\varepsilon$:
\begin{equation}\label{eq:matching}
 k_d+a\varepsilon^\gamma+b\varepsilon^\alpha
 =k_0+k_1\frac{\varepsilon^\gamma}{\varepsilon_0^\gamma+\varepsilon^\gamma}.
\end{equation}
The gain is globally bounded and satisfies $G(r)\ge\underline G:=\min\{k_d,k_0\}$. For each plant, choose $\underline G>D_i$ and define $\delta_i=\underline G-D_i>0$. Besides the
continuity condition \eqref{eq:matching}, the gain has the free parameters
$(k_0,k_1,\epsilon_0,\gamma,\alpha,\epsilon,a,b)$. $k_d$ is then fixed by \eqref{eq:matching}.

\subsection{First-order system}\label{subsec:first}
Substitution of \eqref{eq:first_plant} into \eqref{eq:xi_dot_common} gives the transformed dynamics
\begin{equation}\label{eq:first_xi_open}
 \dot\xi=\chi\left[f_1(x,t)+g_1(x,t)u+d_1(t)-\dot\rho\operatorname{erf}(\xi)\right].
\end{equation}
The prescribed-performance-aware controller is proposed as
\begin{equation}\label{eq:first_control}
 u=\frac{-f_1(x,t)+\dot\rho\operatorname{erf}(\xi)-G(|\xi|)\operatorname{sgn}(\xi)}{g_1(x,t)}.
\end{equation}
For $\xi\ne0$, using \eqref{eq:first_control} its closed-loop dynamics are obtained as 
\begin{equation}\label{eq:first_closed}
 \dot\xi=\chi\left[-G(|\xi|)\operatorname{sgn}(\xi)+d_1(t)\right].
\end{equation}
At $\xi=0$, the discontinuous law is understood in the Filippov sense.

\begin{theorem}\label{th:first}
Under Assumptions 1--2 and \eqref{eq:matching}, let $\underline G>D_1$. The closed loop dynamics \eqref{eq:first_plant}, \eqref{eq:first_control} admits a global forward solution for every initial condition in $\mathcal F(0)$. Every such solution satisfies
\begin{equation}\label{eq:first_norm_bound}
 \frac{|x(t)|}{\rho(t)}\le\operatorname{erf}(|\xi(0)|)<1,
\end{equation}
and reaches $x=0$ in finite time given by
\begin{equation}\label{eq:first_time}
 T_1\le\frac{2\rho_0|\xi(0)|}{\sqrt\pi\,\delta_1}.
\end{equation}
Thereafter the state remains at zero in the ideal Filippov sense.
\end{theorem}
\begin{proof}
We first consider the maximal interval on which $|x|<\rho$ and $\xi$ is finite. For $\xi\ne0$, $W=|\xi|$ satisfies almost everywhere
\begin{align}
 \dot W&=\operatorname{sgn}(\xi)\dot\xi
 \le-\chi\big(G(W)-D_1\big)
 \le-\chi_{\min}\delta_1,\label{eq:first_Wdot}
\end{align}
where $\chi\ge\chi_{\min}:=\sqrt\pi/(2\rho_0)$. Equivalently, with $V_1=\xi^2/2$,
\begin{equation}\label{eq:first_Vdot}
 \dot V_1\le-\chi_{\min}\delta_1\sqrt{2V_1},\qquad \xi\ne0.
\end{equation}
Hence $W$ is nonincreasing prior to reaching zero, and $W(t)\le W(0)$. Inverting \eqref{eq:transform} gives \eqref{eq:first_norm_bound}, so the trajectory stays uniformly separated from the funnel boundary in normalized coordinates. Since the gain, $\dot\rho$, and nominal feedback terms are locally bounded on this compact transformed-state region, no finite-time exit through $|x|=\rho$ can occur. Integrating \eqref{eq:first_Wdot} up to the first time the trajectory reaches $\xi = 0$ yields \eqref{eq:first_time}. At $\xi=0$, the Filippov control set contains the equivalent value required to cancel any $|d_1|\le D_1<\underline G$. Therefore, $x=0$ remains invariant.
\end{proof}
\begin{remark}
For $f_1 = 0$, $g_1 = 1$,~\eqref{eq:first_control} implies
\begin{equation}
u = \frac{\dot{\rho}}{\rho}\,x - \widetilde{G}\!\left(\frac{|x|}{\rho}\right)\mathrm{sgn}(x),
\qquad
\widetilde{G}(y) = G\!\left(\mathrm{erf}^{-1}(y)\right),
\end{equation}
and $y = x/\rho$ obeys
\begin{equation}
\dot{y} = \frac{-\widetilde{G}(|y|)\mathrm{sgn}(y) + d_1}{\rho}.
\end{equation}
The funnel invariance~\eqref{eq:first_norm_bound} therefore follows from a time-scaled
sliding-mode controller on $y$. The erf transformation adds the
explicit reaching-time bound~\eqref{eq:first_time} in terms of $G$ and $\rho_0$,
and a gain $G$ that is defined in the transformed coordinate.
\end{remark}
\begin{prop}[First-order control bound]
Let $\bar{G} := \sup_r G(r)$. Under the conditions of Theorem~\ref{th:first},
\begin{equation}
|u(t)| \leq \frac{\sup|f_1| + \lambda(\rho_0 - \rho_\infty)\,
\frac{|x(0)|}{\rho_0} + \bar{G}}{g_{1,\min}}.
\end{equation}
Conversely, for $f_1 = 0$, $g_1 = 1$, no control with
$|u| < D_1 + \lambda(\rho_0 - \rho_\infty)$ keeps the funnel invariant for
every $|x(0)| < \rho_0$ and every $|d_1| \le D_1$ (
$|x(0)| \to \rho_0$, $d_1 = D_1$ at $t=0$, where the boundary moves at
$-\lambda(\rho_0 - \rho_\infty)$).
\end{prop}

\begin{proof}
We have $|\dot{\rho}| \le \lambda(\rho_0 - \rho_\infty)e^{-\lambda t}$,
$|\mathrm{erf}(\xi)| = |x|/\rho \le |x(0)|/\rho_0$ by~\eqref{eq:funnel},
and $G \le \bar{G}$.
For the converse, at the boundary $x = \rho$ the funnel requires
$\dot{x} \le \dot{\rho}$, i.e., $u + d_1 \le -\lambda(\rho_0 - \rho_\infty)$
for all admissible $d_1$.
\end{proof}
\subsection{Second-order system}\label{subsec:second}
For \eqref{eq:second_plant}, the  transformation \eqref{eq:xi_dot_common} yields
\begin{equation}\label{eq:xi_second_kin}
 \dot\xi=\chi\left(\dot x-\frac{\dot\rho}{\rho}x\right).
\end{equation}
Set $A(t)=\dot\rho/\rho$ and $\dot A=\ddot\rho/\rho-A^2$. Differentiating \eqref{eq:xi_second_kin} gives the exact transformed acceleration equation
\begin{equation}\label{eq:xi_second_accel}
 \ddot\xi=\chi\left[f_2+g_2u+d_2-2A\dot x-(\dot A -A^2) x\right]
 +2\xi\dot\xi^2.
\end{equation}
To avoid direct cancellation of the $2\xi\dot\xi^2$ term, define the virtual velocity and auxiliary error
\begin{align}
 v&=(A-k)x,\qquad z=\dot x-v,\label{eq:z_def}\\
 \dot v&=\dot A x+(A-k)\dot x,\qquad k>0.\label{eq:vdot}
\end{align}
Equations \eqref{eq:xi_second_kin} and \eqref{eq:second_plant} imply
\begin{align}\label{eq:second_cascade}
\begin{split}
 \dot\xi & =\chi(z-kx),\\ \dot{z} & =f_2+g_2u+d_2-\dot v.   
\end{split}
\end{align}
Choosing the control law as
\begin{equation}\label{eq:second_control}
 u=\frac{-f_2+\dot v-\chi\xi-kz-G(|z|)\operatorname{sgn}(z)}{g_2}.
\end{equation}
Using \eqref{eq:second_control} in \eqref{eq:second_cascade} reduces to 
\begin{align}\label{eq:second_closed}
\begin{split}
 \dot\xi& =\chi(z-kx),\\
 \dot z & =-\chi\xi-kz-G(|z|)\operatorname{sgn}(z)+d_2.   
\end{split}
\end{align}

\begin{theorem}\label{th:second}
Under Assumptions 1--2 and \eqref{eq:matching}, suppose $\underline G>D_2$ and $k>0$. For any finite initial velocity and $|x(0)|<\rho_0$, the closed loop dynamics \eqref{eq:second_plant}, \eqref{eq:second_control} has a global forward Filippov solution, satisfies $|x(t)|<\rho(t)$ for all $t\ge0$, and obeys
\begin{equation}\label{eq:second_exp}
 V_2(t)\le V_2(0)e^{-2kt},\quad
 V_2=\tfrac12(\xi^2+z^2).
\end{equation}
Consequently $\xi(t),z(t)\to0$ exponentially and $x(t)\to0$. 
\end{theorem}
\begin{proof}
On the maximal interval on which the transformation is defined, the derivative of $V_2$ along \eqref{eq:second_closed}, for $z\ne0$, is
\begin{align}
 \dot V_2
 &=\xi\chi(z-kx)+z[-\chi\xi-kz-G(|z|)\operatorname{sgn}(z)+d_2]\\
 &=-k\chi\xi x-kz^2-G(|z|)|z|+zd_2\\
 &\le-k\chi\xi x-kz^2-\delta_2|z|.\label{eq:second_Vintermediate}
\end{align}
The inequality $\chi\xi x\ge\xi^2$ follows because, for $r\ge0$,
\begin{equation}\label{eq:erf_ineq}
 \frac{\sqrt\pi}{2}e^{r^2}\operatorname{erf}(r)
 =e^{r^2}\int_0^r e^{-s^2}\,ds\ge r,
\end{equation}
where the last inequality uses $e^{r^2-s^2}\ge1$ on $0\le s\le r$; odd symmetry extends the claim to $\xi<0$. Thus, almost everywhere (including $z=0$),
\begin{equation}\label{eq:second_Vdot}
 \dot V_2\le-k\xi^2-kz^2-\delta_2|z|
 \le-2kV_2.
\end{equation}
Integration yields \eqref{eq:second_exp}. In particular $|\xi(t)|\le\sqrt{2V_2(0)}$ and therefore
\begin{equation}\label{eq:second_norm_bound}
 \frac{|x(t)|}{\rho(t)}
 \le\operatorname{erf}\big(\sqrt{2V_2(0)}\big)<1.
\end{equation}
This strict uniform normalized-error margin prevents exit through the moving funnel boundary. Since $\xi,z$ remain bounded, $x=\rho\operatorname{erf}(\xi)$ is bounded and $\dot x=z+(A-k)x$ is bounded. The controller is consequently bounded along each fixed initial-condition trajectory (under the stated regularity and bounded nominal dynamics on compact sets), and the solution can be continued globally. Exponential convergence of $\xi,z$ follows from \eqref{eq:second_exp} since $\rho$ is bounded above and below, $x=\rho\operatorname{erf}(\xi)$ also converges exponentially.
\end{proof}

\begin{remark}
The term $\chi\,\xi$ in~\eqref{eq:second_control} grows like $e^{\xi^2}$, so the
second-order control is bounded along each trajectory but not uniformly
in the initial condition; practical implementations therefore require
input saturation or a gain shaping of $\chi$ near the funnel boundary.
\end{remark}

\begin{prop}[Actuator requirement]\label{prop:actuator_requirement}
Under the conditions of Theorem~2 let $M=\sqrt{2V_2(0)}$, $|g_2|\ge g_{\min}>0$, $G_{\max}=\sup_{r\ge0}G(r)$, and
\begin{align}
\begin{split}
\bar\chi(M)=\frac{\sqrt\pi\,e^{M^2}} {2\rho_\infty},\quad 
H_M=& \frac{\lambda^2\rho_0}{4}+(\lambda+k)\big(M+(\lambda+k)\rho_0\big),\\ 
F_M=& \sup_{\mathcal K,\,t\ge0}|f_2|.    
\end{split}
\end{align}
where $\mathcal K=\{(x,\dot x):|x|\le\rho_0,\ |\dot x|\le M+(\lambda+k)\rho_0\}$ is compact. Then the control \eqref{eq:second_control} satisfies, along every solution of the unsaturated closed loop,
\begin{equation}
|u(t)|\le U_{\rm req}(M):=\frac{F_M+H_M+\bar\chi(M)M+kM+G_{\max}}{g_{\min}} .
\end{equation}
If $u_{\max}>U_{\rm req}(M)$, every solution of the closed loop with the saturated input $\mathrm{sat}_{u_{\max}}(u)$ is a solution of the unsaturated closed loop, so the conclusions of Theorem~2 hold.
\end{prop}
\begin{proof}
By Theorem 2, $V_2(t)\le V_2(0)$, hence $|\xi|\le M$, $|z|\le M$ and $|x|<\rho\le\rho_0$. Since $\rho\ge\rho_\infty$, $\chi\le\bar\chi(M)$ and $|\chi\xi|\le\bar\chi(M)M$. From \eqref{eq:z_def}, $\dot x=z+(A-k)x$ and $|A|=\lambda(\rho-\rho_\infty)/\rho\le\lambda$, so $|\dot x|\le M+(\lambda+k)\rho_0$ and $(x,\dot x)\in\mathcal K$, which gives the bound $F_M$. Also $\dot A=\ddot\rho/\rho-A^2=\lambda^2(\rho-\rho_\infty)\rho_\infty/\rho^2\in[0,\lambda^2/4]$, so $|\dot v|=|\dot Ax+(A-k)\dot x|\le H_M$. Taking absolute values in \eqref{eq:second_control} and using $|G\,\mathrm{sgn}(z)|\le G_{\max}$ gives the bound. For the saturated loop, let $\tau$ be the largest time up to which a solution solves the unsaturated loop. On $[0,\tau]$ the bound above holds, so the input is $<u_{\max}$ with margin. The non-switching part of \eqref{eq:second_control} is continuous along the solution and the switching part is bounded by $G_{\max}$, so the margin persists on a neighbourhood of $\tau$, contradicting maximality of $\tau$ unless $\tau=\infty$. 
\end{proof}
\emph{Remark.} $U_{\rm req}$ grows like $e^{M^2}$ and is very conservative: for $(x,\dot x)(0)=(2,-0.3)$, $k=1$ one has $M=4.21$ and $\bar\chi M\approx5\cdot10^{8}$ whereas the simulated peak is $7.08$. For $(2.3,-0.3)$ the bound is $\approx4\cdot10^{11}$. It shows that the requirement is not uniform in the initial condition.

\section{Simulation Study}\label{sec:simulations}
This section compares the proposed controller with the non-PPF
hybrid-gain controller of~\cite{shivamArxiv2024} and a constant-gain-based conventional sliding-mode
controller. Both plants are subject to $d(t) = 0.25\sin(10t)$ and are
integrated by explicit Euler with step $10^{-4}$~s; every discontinuous
term uses the ideal $\mathrm{sgn}(\cdot)$. Two scenarios are used:
(a)~baselines tuned, by a grid over gain scale (and, in second order, the
surface slope $c$), to the smallest gains that keep them inside the
funnel at the nominal initial condition, and
(b)~all controllers limited to the same peak control.
Robustness is assessed by keeping the nominal tuning while changing the
initial condition. Metrics: $T_{0.02}$, the last time $|x| > 0.02$;
$\mathrm{IAE} = \int |x|\,dt$; $\max|u|$; the control effort
$\int u^2\,dt$ , split at $T_s$ into
$J_u^{\mathrm{tr}}$ ($t < T_s$) and $J_u^{\mathrm{ss}}$ ($t \ge T_s$),
with $T_s = 1$~s (first order) and $5$~s (second order), so that the
switching energy of the ideal $\mathrm{sgn}$ after convergence is
separated from the energy of the transient; and $\max|x|/\rho$, where
values above~$1$ mean that the funnel is left.

\subsection{First-order system}\label{subsec:sim_first}

For $\dot x=u+d$, the performance function is $\rho(t)=(4-0.05)e^{-4t}+0.05$, with $x(0)=3.5$ and $T=6\,\mathrm{s}$. The proposed hybrid-gain parameters are listed in Table~\ref{tab:params hybrid gain}. 
\begin{table}[!t]
\caption{Hybrid-gain parameters used in Section~\ref{sec:simulations}.}
\label{tab:params hybrid gain}
\centering
\renewcommand{\arraystretch}{1.15}
\setlength{\tabcolsep}{5pt}
\begin{tabular}{lccccccccc}
\toprule
System & $k_0$ & $k_1$ & $\varepsilon_0$ & $\gamma$ & $\alpha$
 & $\varepsilon$ & $a$ & $b$ & $k_d$ \\
\midrule
First order   & 0.6   & 1.9   & 0.6   & 0.7   & 1.5   & 0.2   & 0.2   & 0.5   & 1.092 \\
Second order  & 0.8   & 1.6   & 0.3   & 0.7   & 1.5   & 0.2   & 1.674 & 0.5   & 0.9   \\
\bottomrule
\end{tabular}
\par\smallskip
\begin{minipage}{\columnwidth}
\footnotesize Baselines scaled by $s$ multiply $k_0$, $k_1$, $a$, $b$, $k_d$ by $s$ ($k_d$ follows from~(9)).
\end{minipage}
\end{table}
\begin{figure}
    \begin{subfigure}[b]{.25\textwidth}
			 \centering	
          \includegraphics[width=\linewidth]{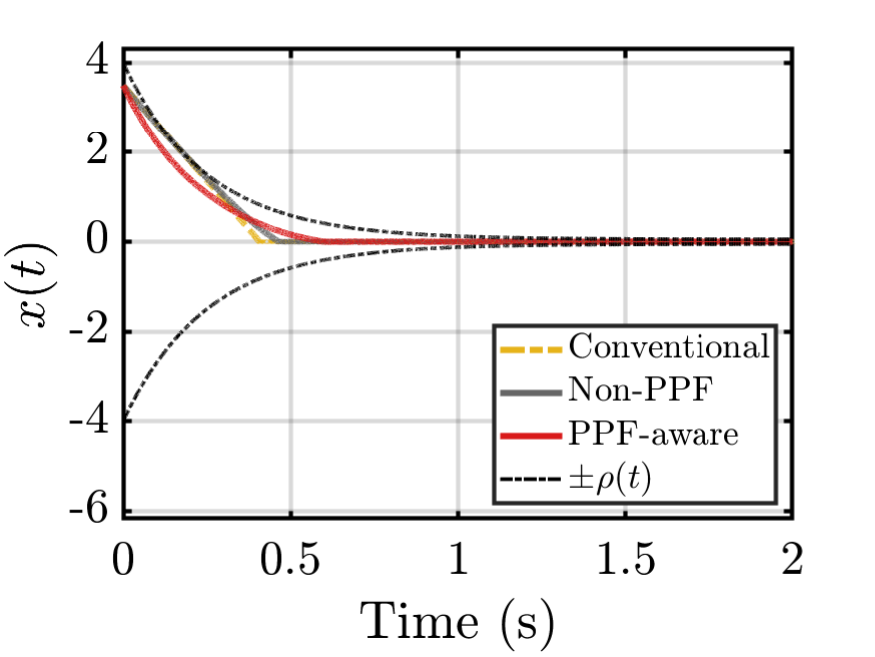}
			\caption{ State trajectories}
			\label{fig: state fos}     
		\end{subfigure}%
		\begin{subfigure}[b]{.25\textwidth}
			\centering		\includegraphics[width=\linewidth]{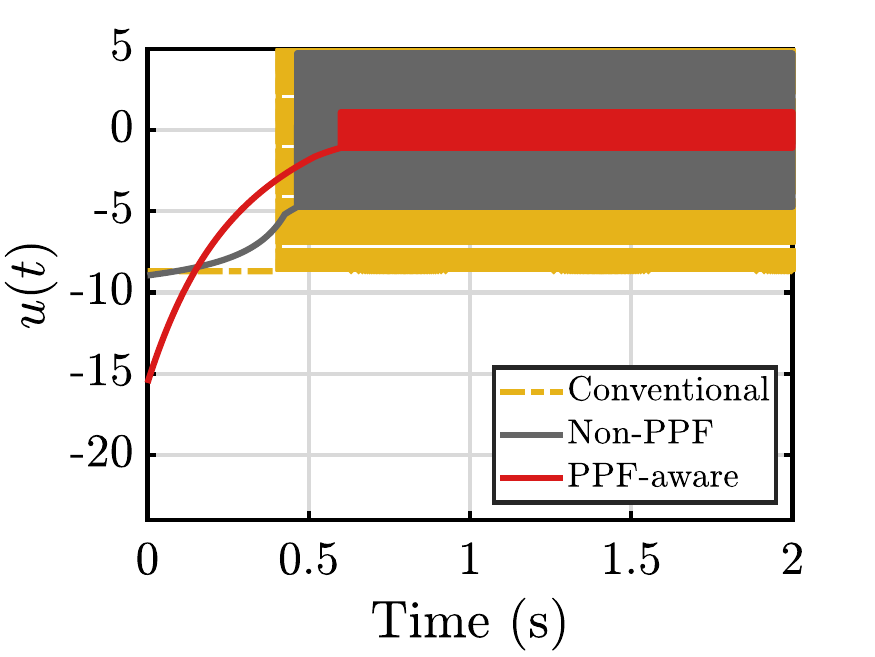}
			\caption{ Control profiles}
			\label{fig: control fos}
		\end{subfigure}%
        \qquad 
        \begin{subfigure}[b]{.45\textwidth}
			\centering		\includegraphics[width=\linewidth]{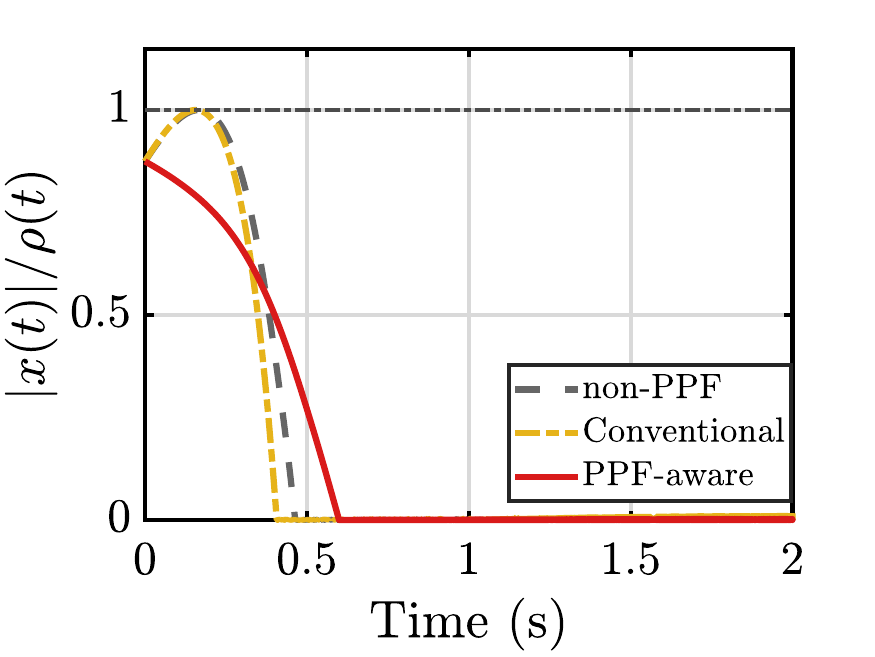}
			\caption{ Normalised error }
			\label{fig: Normalised error}
		\end{subfigure}%
        \caption{ Results for first order system}
        \label{fig: Results for first order system}
        \end{figure}
Fig. \ref{fig: Results for first order system} illustrates the simulation results for the upper block in Table \ref{tab:first_order} where the baselines are tuned to remain in the funnel. Fig. \ref{fig: state fos} shows that the state trajectories converge monotonically to the origin for all methods. The corresponding control inputs shown in Fig. \ref{fig: control fos} exhibit comparable peak magnitudes, confirming uniform robustness of the design. 
\begin{table}[t]
\centering
\caption{First-order comparison results at $x(0)=3.5$. The upper block uses baselines tuned to converge at least as fast as the proposed controller while remaining inside the funnel. The lower block uses an equal peak-control amplitude($\max|u|\leq15.57$) with the lowest-IAE feasible. Baselines in the lower block are unsaturated.}
\label{tab:first_order}
\setlength{\tabcolsep}{3.0pt}
\begin{tabular}{llrrrrr}\hline
Controller & Tuning & $T_{0.02}$ & IAE & $\max|u|$ & $J_u^{\rm tr}$ & $J_u^{\rm ss}$\\\hline
\multicolumn{7}{l}{\emph{Baselines tuned to remain in the funnel}}\\
PPF-aware & none & 0.58 & 0.665 & 15.6 & 30.5 & 6.0\\
Non-PPF  & gain $\times4.32$ & 0.46 & 0.750 & 8.9 & 39.4 & 111.4\\
Conventional & gain $\times3.48$ & 0.40 & 0.719 & 8.7 & 75.7 & 378.4\\
\hline
\multicolumn{7}{l}{\emph{Equal peak control ($\max|u|\le15.57$)}}\\
PPF-aware & none & 0.58 & 0.665 & 15.6 & 30.5 & 6.0\\
Non-PPF  & gain $\times7.50$ & 0.27 & 0.432 & 15.5 & 96.9 & 336.0\\
Conventional & gain $\times6.12$ & 0.23 & 0.409 & 15.3 & 234.1 & 1170.5\\
\hline\end{tabular}
\end{table}
\begin{table}[t]
\centering
\caption{First-order robustness over $x(0)\in\{1,2,3,3.5,3.9\}$, without saturation. Baselines retain their nominal tuning.}
\label{tab:first_robust}
\setlength{\tabcolsep}{20pt}
\begin{tabular}{@{}lcc@{}}
\toprule
Controller & Violations & Worst $\max|x|/\rho$\\
\midrule
PPF aware & 0/5 & 0.97\\
Non-PPF & 1/5 & 1.19\\
Conventional & 1/5 & 1.20\\
\bottomrule
\end{tabular}
\end{table}
With baselines tuned to the smallest gains that keep them inside the
funnel (upper block of Table~\ref{tab:first_order}), the proposed law
reaches the $0.02$-neighbourhood in $0.58$~s against $0.46$~s and
$0.40$~s, with similar IAE ($0.665$ against $0.750$ and $0.719$). Its
transient energy is lower ($J_u^{\mathrm{tr}} = 30.5$ against $39.4$
and $75.7$). Most of the additional energy of the baselines is in the
steady-state window ($111.4$ and $378.4$ against $6.0$). Following
convergence the ideal $\mathrm{sgn}$ switches with the gain at the
origin, due to which the uniform gain scaling ($\times 4.32$ and
$\times 3.48$) has increased, which is an effect of the tuning rule and of
the ideal sign function. With equal peak control
(lower block) the baselines uses control on a faster response
($T_{0.02} = 0.27$~s and $0.23$~s) and obtain lower IAE ($0.432$
and $0.409$) at a transient energy of $96.9$ and $234.1$, and they
also stay inside the funnel at $x(0) = 3.5$. The proposed first-order
law therefore shows funnel
invariance without retuning, with an explicit reaching-time bound, whereas it does not improve tracking.
Table~\ref{tab:first_robust} shows results for the minimum-gain baselines,
which leave the funnel at $x(0) = 3.9$ ($1.195$ and $1.203$) while the
proposed law does not ($0.975$). With the larger gains of the
equal-peak block, the baselines also remain in the funnel for all five
initial conditions, so the advantage is against baselines tuned at the
nominal condition with minimal gains.

\subsection{Second-order system}\label{subsec:sim_second}

The second-order simulation study uses $\ddot x=-2\zeta\omega_n\dot x-\omega_n^2x+u+d$, with $\omega_n=2$, $\zeta=0.15$, $(x(0),\dot x(0))=(2,-0.3)$, and $T=10\,\mathrm{s}$. The prescribed envelope has $(\rho_0,\rho_\infty,\lambda)=(2.5,0.35,1.4)$. The proposed controller employs hybrid gain parameters as listed in Table~\ref{tab:params hybrid gain} and $k_d=0.9$, with auxiliary velocity error $z=\dot x-(\dot\rho/\rho-k)x$. The non-PPF controllers use $s=\dot x+cx$, with both the surface slope and gain scale included in the tuning search. The non-PPF controller parameter setting $c=0.8$ is listed separately from the retuned hybrid-gain baseline.

Figure~\ref{fig: numerical simulation results for second order} illustrates the second-order
closed-loop responses under the nominal initial condition.
As shown in Fig.~\ref{fig: Position profile}, all controllers
drive the position toward the origin, although their
transient responses differ. The corresponding velocity
profiles in Fig.~\ref{fig: Velocity profile} show the
initial corrective action and subsequent convergence.
Notably, Fig.~\ref{fig: Normalised error} reveals
that the reference non-PPF controller exceeds the
prescribed-performance boundary, whereas the proposed
PPF-aware controller and the tuned baselines remain
within the envelope.
The control histories in
Figs.~\ref{fig: control profile funnel second order} and
\ref{fig: control profile peak second order} illustrate a more pronounced
difference wherein the tuned non-PPF and conventional
controllers exhibit persistent high-frequency
switching, while the proposed controller maintains
a substantially smaller steady-state control
amplitude. This behaviour is reflected in the
steady-state squared-control effort reported in
Table~\ref{tab:second_order}, where the proposed
controller achieves $J_{u,\mathrm{ss}}=4.1$,
compared with $65.4$ and $185.9$ for the tuned
non-PPF and conventional controllers under the
common actuator limit. 
Similar to the previous case study
Table~\ref{tab:second_order} compares the controllers for two protocols:
baselines tuned to stay in the funnel with all controllers saturated at
$6.097$, and all controllers limited to $7.08$. At $6.097$ the proposed
controller settles in $3.29$~s with $\max|x|/\rho = 0.960$. The
reference non-PPF controller ($c = 0.8$) needs $6.08$~s and leaves the
funnel ($1.320$). Restoring compliance for the non-PPF controller
needed both a steeper surface and a larger gain $c = 1.5$ with gain
$\times 4$ settles in $3.20$~s and stays inside, whereas $c = 1.5$ at
$\times 1$ or $c = 0.8$ at $\times 4$ do not ($1.361$ and $1.037$). The steady-state energy of the proposed controller
($J_u^{\mathrm{ss}} = 4.1$) is $94\%$ and $98\%$ below that of the
tuned non-PPF and constant-gain controllers ($65.4$ and $185.9$). But
the reference non-PPF controller with unscaled gains has the same
$4.1$, so this reduction is relative to gain-scaled baselines and
reflects the switching energy of the ideal $\mathrm{sgn}$. With the
common limit $7.08$, the proposed controller, which is not clipped,
settles in $3.31$~s with $\mathrm{IAE} = 1.573$. The tuned non-PPF and
constant-gain controllers ($\times 3$ and $\times 2.5$) settle in
$3.29$~s with IAE $1.745$ and $1.746$ and $J_u^{\mathrm{tr}}$ of
$52.6$ and $167.6$ against $21.8$. The reference non-PPF controller
leaves the funnel with a peak of $6.58$, below the limit, so raising
the admissible peak alone does not ensure funnel compliance for a fixed
tuning.

Robustness study in Table~\ref{tab:second_robust} is carried out without control saturation. Therein the
proposed controller keeps $|x| < \rho$ for all six initial conditions
(worst $0.999$, at $(2.3, -0.3)$, where the unsaturated peak is $231$).
The reference, tuned non-PPF and tuned constant-gain controllers
violate it in five, three and one cases. With $|u| \le 6.097$ the
proposed controller leaves the funnel by at most $2\%$ in two cases,
against five, three and one for the three baselines whereas with the $7.08$
limit the numbers are two (worst $1.022$), five ($1.835$), three
($1.145$) and three ($1.155$).

\begin{figure}
    \begin{subfigure}[b]{.25\textwidth}
			 \centering
        \includegraphics[width=\linewidth]{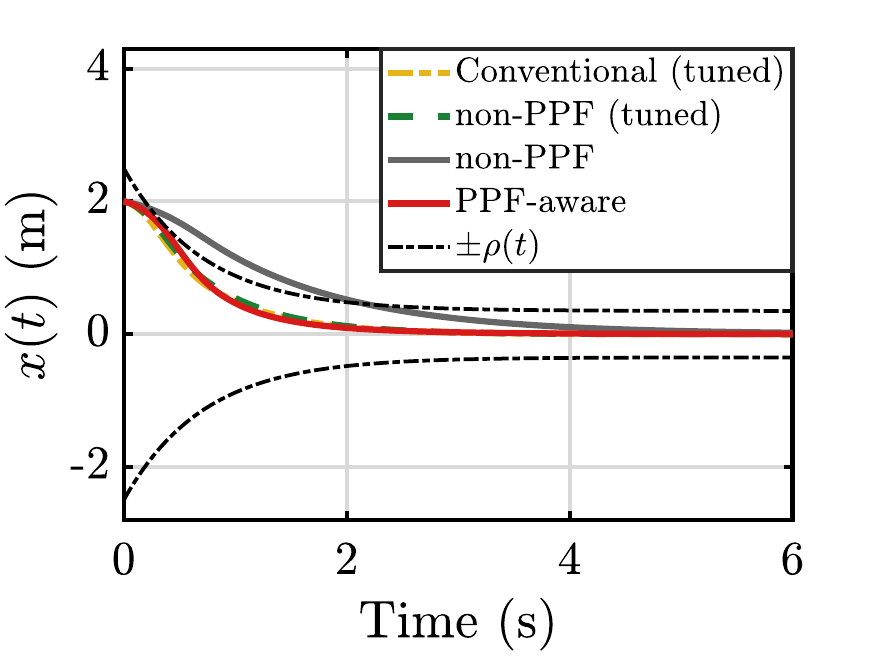}
			\caption{ Position profiles}
			\label{fig: Position profile}     
		\end{subfigure}%
		\begin{subfigure}[b]{.25\textwidth}
			  \centering	
         \includegraphics[width=\linewidth]{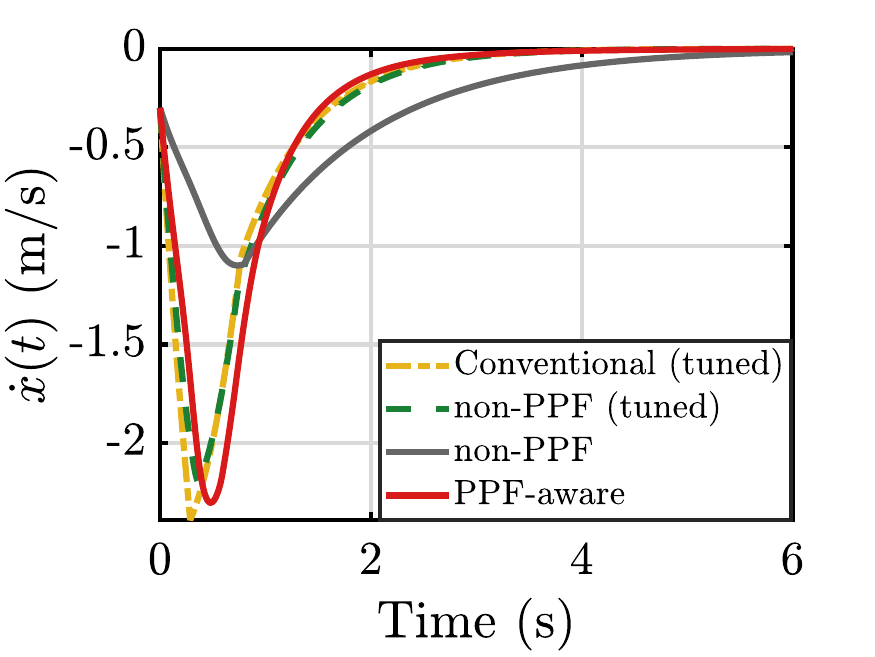}
			\caption{ Velocity profiles }
			\label{fig: Velocity profile}
		\end{subfigure}%
        \qquad
         \begin{subfigure}[b]{.3\textwidth}
			\centering
           \includegraphics[width=\linewidth]{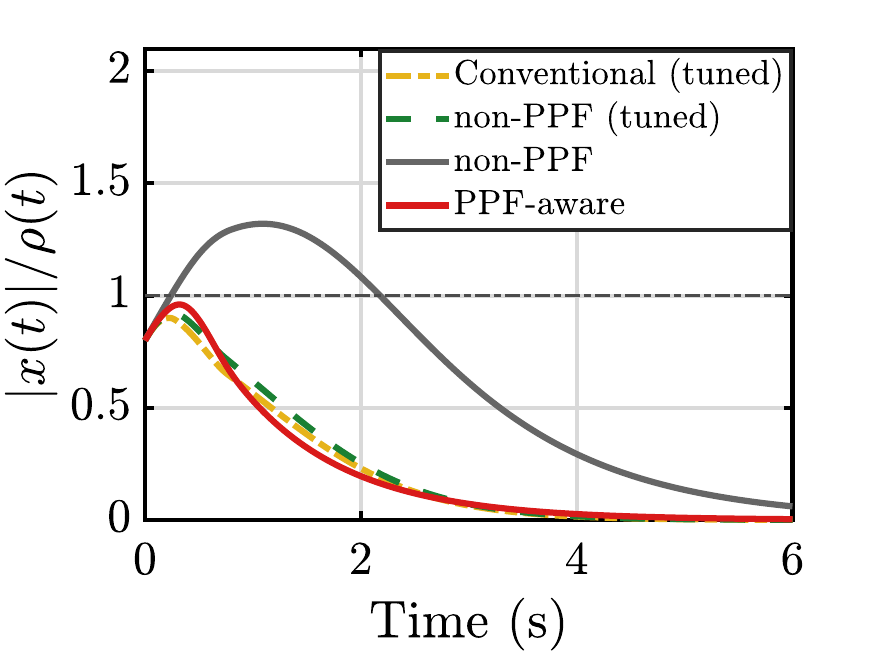}
			\caption{ Normalised errors}
			\label{fig: Normalised error}     
		\end{subfigure}%
         \qquad
         \centering
		\begin{subfigure}[b]{.3\textwidth}

         \includegraphics[width=\linewidth]{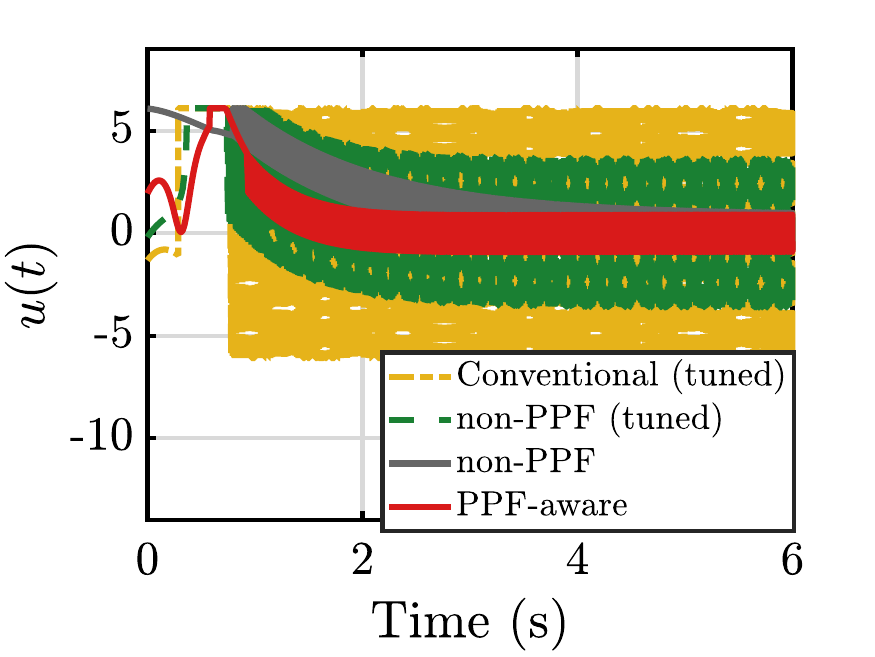}
			\caption {Control profiles over time while in the funnel}
			\label{fig: control profile funnel second order}
		\end{subfigure}%
        \qquad
        \centering
        \begin{subfigure}[b]{.3\textwidth}		
            \includegraphics[width=\linewidth]{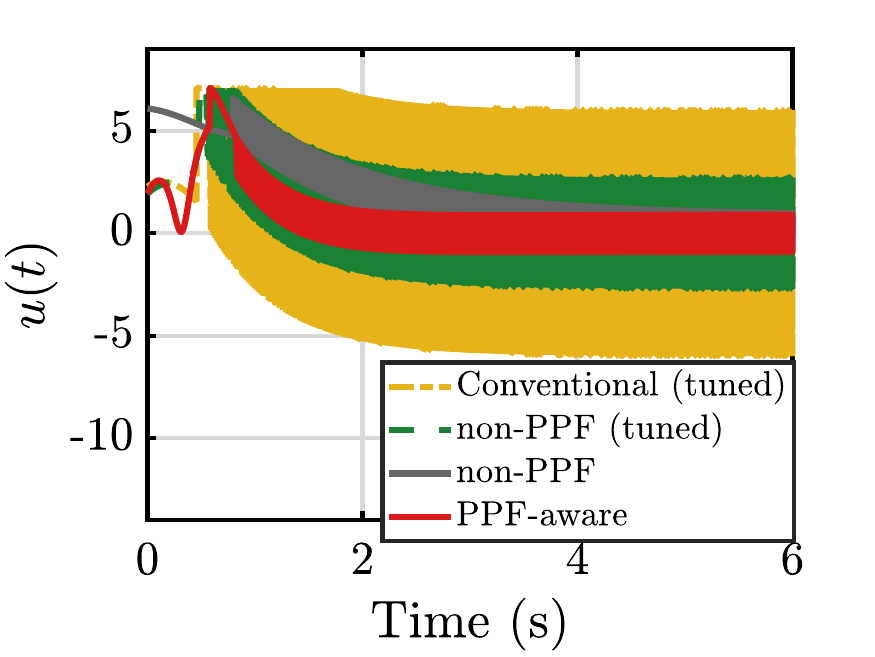}
			\caption{ Control profiles over time with equal peak }
			\label{fig: control profile peak second order}
		\end{subfigure}%
         \caption{ Results for second-order system }
			\label{fig: numerical simulation results for second order}
\end{figure}

\begin{table}
\centering
\caption{Second-order results at $(x,\dot{x})(0) = (2,-0.3)$.
Upper block: baselines tuned to the smallest gains that keep them inside
the funnel, all controllers saturated at $6.097$. Lower block: all
controllers saturated at $7.08$ (equal peak control. The proposed
controller is not clipped), tuned baselines as above. $J_u^{\mathrm{tr}}$,
$J_u^{\mathrm{ss}}$: control energy for $t < 5$~s and $t \ge 5$~s.}
\label{tab:second_order}
\setlength{\tabcolsep}{2pt}
\begin{tabular}{l c c c c c c}
\toprule
Controller & $T_{0.02}$ & IAE & $\max|u|$ & $J_{u,\mathrm{tr}}$ &
$J_{u,\mathrm{ss}}$ & $\max|x|/\rho$\\
\midrule
\multicolumn{7}{l}{\emph{Baselines tuned to remain in the funnel (saturated at 6.097)}}\\
PPF-aware ($k{=}1$)              & 3.29 & 1.564 & 6.10 & 21.5  & 4.1   & 0.960\\
non-PPF ($c{=}0.8$, $\times1$)   & 6.08 & 3.090 & 6.10 & 47.3  & 4.1   & 1.320\\
non-PPF ($c{=}1.5$, $\times4$)   & 3.20 & 1.609 & 6.10 & 75.3  & 65.4  & 0.920\\
Conventional ($c{=}1.5$, $\times4$) & 3.14 & 1.524 & 6.10 & 175.5 & 185.9 & 0.901\\
\midrule
\multicolumn{7}{l}{\emph{Equal peak control $\bar u = 7.08$}}\\
PPF-aware ($k{=}1$)              & 3.31 & 1.573 & 7.08 & 21.8  & 4.1   & 0.960\\
non-PPF ($c{=}0.8$, $\times1$)   & 6.08 & 3.090 & 6.58 & 47.4  & 4.1   & 1.320\\
non-PPF ($c{=}1.5$, $\times3$)   & 3.29 & 1.745 & 7.08 & 52.6  & 36.7  & 0.964\\
Conventional ($c{=}1.5$, $\times2.5$) & 3.29 & 1.746 & 7.08 & 167.6 & 180.0 & 0.969\\
\bottomrule
\end{tabular}
\end{table}

\begin{table}
\centering
\caption{Robustness over six initial conditions.
Non-PPF = reference tuning ($c = 0.8$, $\times 1$);
worst = largest $\max|x|/\rho$ over the six initial conditions
(the proposed controller without saturation $0.999$).}
\label{tab:second_robust}
\setlength{\tabcolsep}{10pt}
\begin{tabular}{@{}lcccc@{}}
\toprule
& \multicolumn{2}{c}{No saturation} & \multicolumn{2}{c}{$|u|\leq6.097$}\\
\cmidrule(lr){2-3}\cmidrule(l){4-5}
Controller & Violations & Worst & Violations & Worst\\
\midrule
PPF-aware & 0/6 & 1.00 & 2/6 & 1.02\\
Non-PPF  & 5/6 & 1.83 & 5/6 & 1.81\\
Non-PPF tuned & 3/6 & 1.08 & 3/6 & 1.08\\
Conventional tuned & 1/6 & 1.06 & 1/6 & 1.06\\
\bottomrule
\end{tabular}
\end{table}

\section{Conclusions}\label{sec:conclusion}

A prescribed-performance-aware hybrid-gain robust
controller was developed for first- and second-order
nonlinear systems with bounded matched disturbances.
The proposed transformation establishes explicit
transient-error constraints, with finite-time convergence
proved for first-order dynamics and exponential
convergence for second-order dynamics. For second-order systems, a virtual-velocity design
guarantees exponential convergence, provided the
actuator does not saturate. Both results
establish funnel invariance without assuming a priori
boundedness of the transformed error. Numerical simulation results compared the controllers with a
non-PPF hybrid-gain controller and a constant-gain sliding-mode
controller, tuned either to stay in the funnel or to the same peak
control. In first order, the proposed law enforces the funnel guarantee without
retuning and a lower transient control energy, whereas the tuned baselines track better at the cost of larger steady-state
energy which comes from their gain scaling under the
ideal $\mathrm{sgn}$. In second order, at a common limit of $7.08$, the
proposed controller reaches lower IAE ($1.573$ against $1.745$ and
$1.746$) and lower control energy, while the reference non-PPF
controller leaves the funnel ($1.320$). Under actuator saturation, the
guarantee does not hold in general. Future work will
characterise the actuator-feasible set of initial states, with uncertainties and drift in the considered system dynamics.

\bibliographystyle{unsrt}   %
\bibliography{refICC2025}

\end{document}